\documentclass{article}

\usepackage[utf8]{inputenc} 
\usepackage[T1]{fontenc}    
\usepackage{url}            
\usepackage{booktabs}       
\usepackage{amsfonts}       
\usepackage{nicefrac}       
\usepackage{microtype}      

\usepackage[skip=0pt]{caption}

\usepackage{graphicx}
\usepackage[dvipsnames]{xcolor}
\usepackage[colorlinks,
urlcolor=ForestGreen,
linkbordercolor=red,
citecolor=red]{hyperref}
\hypersetup{
	colorlinks = true
}

\usepackage{amssymb}      
\usepackage{amsthm}       
\usepackage{commath}      
\usepackage{amsmath}      
\usepackage{amsfonts}     
\usepackage{bm}           
\usepackage{xargs}        
\usepackage{float}
\usepackage{tikz}

\usepackage{colortbl}  
\definecolor{highlightyellow}{rgb}{1.0, 1.0, 0.75}

\usepackage[nonatbib,preprint]{neurips_2026}
\usepackage[
    backend=biber,
    style=numeric-comp,  
    sorting=none,        
    giveninits=true,     
    maxbibnames=2,       
    minbibnames=1,
    doi=true,           
    isbn=false,
    url=false,
    eprint=false,
]{biblatex}

\renewbibmacro{in:}{%
  \ifentrytype{article}{}{\printtext{\bibstring{in}\intitlepunct}}}

\AtEveryBibitem{\clearlist{language}\clearfield{langid}}

\renewbibmacro*{name:andothers}{%
  \ifboolexpr{
    test {\ifnumequal{\value{listcount}}{\value{liststop}}}
    and
    test \ifmorenames
  }
    {\andothersdelim\bibstring[\textit]{andothers}} 
    {}}

\newcommand{\paren}[1]{\left(#1\right)}

\theoremstyle{plain}

\newtheorem{proposition}{Proposition}

\theoremstyle{definition}

\theoremstyle{remark}

\numberwithin{proposition}{section}
\numberwithin{equation}{section}

\newcommand{\MatNormal}[5]{\mathcal{MN}_{#1\times #2}\left(#3; #4, #5\right)}

\newcommandx{\E}[2][1=]{\mathbb{E}_{#1}\left(#2\right)} 

\newcommand{\tr}{\operatorname{tr}}  
\newcommand{\logdet}{\log\det}       

\newcommand{\fronorm}[1]{\|#1\|_F}   

\newcommand{\est}[1]{\widehat{#1}}   

\newcommand{\reals}{\mathbb{R}}      
\newcommand{\Spp}[1]{\mathbb{S}_{++}^{#1}}  

\newcommand{\KL}[2]{\mathrm{KL}\!\left(#1 \;\|\; #2\right)}

\title{Partitioning Neural Co-Variability}

\author{%
  Skyler Thomas \\
  Johns Hopkins University\\ Department of Biomedical Engineering\\
  \texttt{sthom215@jhu.edu} \\
  \And
  Brandon J. Zhu \\
  Johns Hopkins University\\ Department of Biomedical Engineering \\
  \texttt{bzhu18@jhmi.edu} \\
  \And
  Kathleen E. Cullen \\
  Johns Hopkins University\\ Department of Biomedical Engineering \\
  Center for Hearing and Balance\\
  Kavli Neuroscience Discovery Institute\\
  \texttt{kathleen.cullen@jhu.edu} \\
  \And
  Adam S. Charles \\
  Johns Hopkins University\\ Department of Biomedical Engineering \\
  Center for Imaging Science\\
  Kavli Neuroscience Discovery Institute\\
  \texttt{adamsc@jhu.edu}
}

\begin{document}

\maketitle

\begin{abstract}
Trial-to-trial variability of neural responses has been linked to important aspects of neural computation and is essential for understanding how neuronal populations respond. While current overdispersion models treat each neuron's gain as independent of each other, this assumption fails to capture the network statistics of neuronal populations. As no existing model can capture overdispersed structured spiking gain-modulation across a neural population, network-level gain covariance remains largely unstudied. We thus present the Poisson matrix-normal latent variable (PMNLV) model, which extends single-neuron overdispersion to neural populations by placing a matrix-normal prior over the latent gain with a Kronecker-factored covariance. Spike counts are Poisson-distributed with a rate equal to the sum of a per-neuron stimulus tuning term and the matrix-normal gain, passed through a quadratic soft-rectifying link. We derive two complementary estimation algorithms: a variational EM (VEM) with a matrix-normal posterior that recovers dense Kronecker factors without structural assumptions, and a Kernel Tournament Method (KTM) that performs data-driven selection over a biologically motivated kernel dictionary and composite likelihood. On simulated data, both algorithms recover the inter-neuron and temporal covariance factors alongside accurate tuning curves. Applying VEM to Neuropixel recordings across four cortical regions of the mouse visual hierarchy, we replicate a previous finding that single-neuron marginal variability changes little across cortical areas. We then show that shared population co-variability, invisible to scalar summaries such as the Fano factor, peaks in primary visual cortex and declines in higher visual areas, consistent with functional connectivity studies. The PMNLV framework is applicable to any simultaneously recorded population where structured gain covariance is of scientific interest. We present methods for basic neuroscience research that elucidate population-level changes without compromising individual neuronal statistics.
\end{abstract}

\section{Introduction}\label{sec:introduction}

Neural coding, or how neural responses change as a function of external stimuli or behavior, forms a fundamental language for how neuroscience quantifies information processing in the brain. A classical observation has been that beyond firing pattern changes in response to different stimuli, neurons also exhibit significant variability between presentations of the same stimuli. Modeling this trial-to-trial variability has proven important for creating better models of neural activity that appropriately capture the conditional distribution of response observations~\cite{keeley2020, churchland2011,dichter2016,lakshminarasimhan_dynamical_2023}. 

Models of neural response variability include the Fano factor (ratio of the response variance to the mean), doubly stochastic models~\cite{goris2014, charles2018, aghamohammadi2025}, and renewal processes that characterize variability via inter-spike-interval (ISI) statistics~\cite{koyama2015, cunningham2007, adams2009}. Within this body of work, much of the literature has focused on the idea of over-dispersed Poisson models, i.e., models that overlay additional sources of variability within the traditional Poisson characterization of the number of observed spikes from a neuron within a given time-bin. These models focus on the fact that most cortical systems exhibit statistics where the variance is greater than the mean~\cite{gur_response_1997, barberini_comparison_2001,friedenberger2023}, and build off of the traditional Poisson model via a hierarchical model. Overdispersion models capture changing response statistics within and across trials~\cite{goris2014,charles2018}, but consider only single neurons, ignoring population-level coordination in trial-to-trial variability gains. 

Accurately modeling correlated gains is important, as they reflect large-scale projections into the population~\cite{carandini2004}, providing insight into neural processing beyond the recorded spikes. Neural populations exhibit qualitatively new statistical structure absent at the level of individual neurons~\cite{anderson_more_1972}, and among these emergent properties, the covariance of gain fluctuations across neurons and time determines how much information a population transmits to downstream areas~\cite{moreno-bote2014}. To bridge this gap, we propose the Poisson matrix-normal latent variable (PMNLV) model that relaxes the typical assumption taken in single-neuron overdispersion models where the residual gain fluctuations are statistically independent across neurons. This extension extends the traditional Poisson model of neural spike counts by introducing a neuron-specific latent gain that modulates the firing rate~\cite{goris2014,charles2018,rupasinghe2025}, with more recent extensions partitioning variability into rate and timing components via renewal processes~\cite{aghamohammadi2025}. These models accurately characterize single-neuron statistics but treat each neuron's gain as independent of every other. Population models such as GPFA and GP-LVM capture smooth shared variability across neuronal activity through a low-dimensional latent trajectory~\cite{lawrence2005,yu2008a}, but the observation model is Gaussian and often assumes independent noise, and neither model accounts for overdispersion in the marginal spike count distributions. No existing model simultaneously permits correlated gain fluctuations across neurons and time while respecting the Poisson spiking statistics of the marginal observations.

In our proposed PMNLV population-level overdispersion model, we model spike counts as Poisson distributed with a matrix-normal latent gain $\mathcal{MN}(\bm{0}, \bm{U}, \bm{V})$, i.e., the gain across neurons and time (or trial) have a Kronecker factored covariance $
\mathbf{\Sigma} = \mathbf{V} \otimes \mathbf{U}$. The Kronecker factors encode inter-neuron gain correlations in $\mathbf{U} \in \mathbb{R}^{N \times N}$ and temporal covariance across trials in $\mathbf{V} \in \mathbb{R}^{T \times T}$. We derive two complementary estimation algorithms: a variational EM (VEM) algorithm that jointly infers tuning curves and dense factors $\bm{U}, \bm{V}$ without structural assumptions on the latent gain covariance, and a Kernel Tournament Method (KTM) that performs data-driven selection over a biologically motivated kernel dictionary via pairwise composite likelihood. By applying VEM to Neuropixels recordings across four cortical regions, we find that in addition to replicating previous results that marginal single-neuron variability is near flat across visual hierarchy, shared population co-variability actually peaks in primary visual area (VISp) and subsequently declines in higher areas. This result illustrates how VEM can detect structured gain covariance in large neural data sets that are invisible to the single-unit variability models.

\section{Background}\label{sec:background}
We review three bodies of work that motivate the PMNLV model: (i) single-neuron overdispersion models, which establish that trial-to-trial variability is dominated by latent gain fluctuations but treat the gain of each neuron as statistically independent of every other; (ii) population models, which capture shared structure across neurons but impose conditionally independent Gaussian observations, absorbing correlated gain fluctuations into residual noise; and (iii) other models e.g., coupled spike-generation models, which encode pairwise dependencies through the spike-generation mechanism but lack a latent variable that co-modulates gain across the network. No existing model simultaneously accounts for correlated gain fluctuations across neurons and time while respecting the Poisson marginal statistics of spike counts.

\paragraph{Single-neuron models of neural variability:}
The Poisson model has been used extensively in neuroscience to study spike count variability and treats observed counts as Poisson-distributed with a stimulus-dependent rate $\lambda_i$~\cite{zemel_probabilistic_1998, nirenberg2003, amarasingham_spike_2006, cohen2011}.
\begin{equation}
	y_i \mid \lambda_i \sim \mathrm{Poisson}(g(\lambda_i)),
\end{equation}
where $g(\cdot)$ is the inverse link function and the underlying rate $\lambda_i$ is often considered as a function of the task, e.g., the input stimuli: $\lambda_i = f_i(x)$. However, it is widely known from empirical observations that a fundamental property of Poisson counts, i.e., that the mean equals the variance  ($\mathrm{Var}(y_i) = \mathbb{E}[y_i]$), does not hold~\cite{gur_response_1997, barberini_comparison_2001,friedenberger2023}.

\paragraph{Gain-modulated Poisson models:}
More recent directions have extended the probabilistic model to account for trial-by-trial gain fluctuations. Generally, these approaches often take the form of a hierarchical model where the latent Poisson rate driving the spike counts is itself random, providing ``super-Poisson'' properties. Specifically, we can consider the rate $\lambda_i$ as itself being a function of a random variable $n_i$ in addition to the stimuli,
\begin{equation}
	y_i \mid n_i \sim \mathrm{Poisson}\!\left(g(f_i(x) + n_i)\right).
\end{equation}
For example, in~\cite{goris2014} a product factor $h(f_i(x),n_i) = n_i g(f_i(x))$ between a Gamma-distributed gain $n_i$ and the stimulus-driven rate $\gamma_i = g(f_i(x))$ was used, with mean and variance $\mathbb{E}[y_i] = \gamma_i$ and $\mathrm{Var}(y_i) = \gamma_i + \gamma_i^2 \mathrm{Var}(n_i)$, giving a quadratic relationship between model variance and mean as the stimulus varies. A more general version in~\cite{charles2018} considered the additive-nonlinear form $g(f_i(x) + n_i)$ where a Gaussian-distributed gain $n_i$ is added to the stimulus-driven rate before the link function $g(\cdot)$, producing a more flexible model of the variance as a function of the mean.

\paragraph{Renewal processes:}
An alternative way to model single-neuron variability is via queuing theory where firing rate variability is modeled via both the relative times between events and a time-varying density of spikes, rather than the overall counts in a time-bin~\cite{cunningham2007,koyama2015,aghamohammadi2025,adams2009}. In this setting, a Poisson distribution is modeled by events that occur with independent, exponentially distributed inter-spike intervals (ISIs) with a mean related to the rate $\lambda_i$, and over-dispersion can be achieved via modifications to the ISI distribution. For example, recent work in~\cite{aghamohammadi2025} modeled spike occurrences as a doubly stochastic renewal point process defined by the pair $\{\tau(\cdot), \lambda(t)\}$, where $\tau(\cdot)$ is the ISI distribution in operational time and $\lambda(t)$ is the instantaneous firing rate. Spikes are first generated in operational time $t'$ by sampling ISIs from $\tau(\cdot)$, with the mean ISI fixed to $\mu_\tau = 1\,\mathrm{s}$, i.e., the firing rate in operational time is unity. Spike times are then mapped to real time $t$ via the inverse cumulative firing rate, $t = \Lambda^{-1}(t')$, or $t' = \Lambda(t) = \int_0^t \lambda(s)\,ds$. Across all single-neuron formulations, the gain $n_i$ is treated as statistically independent across neurons while the population-level covariance structure of these fluctuations remains unmodeled.

\paragraph{Population models of neural activity:} Complementing the literature that seeks improved models of each neuron's response characteristics has been the effort to model the interactions between neurons. Population-level analyses of firing rates generally fall into two classes. The first are methods that seek a low-dimensional latent trajectory of neural activity, e.g., in a latent space given by probabilistic principal component analysis or other methods e.g., neuron-coupling ~\cite{pillow2008} or neural networks~\cite{sussillo2016}.

\paragraph{Latent population models:} Population-level models often focus on describing the shared geometry underlying the neural activity across the network. One common geometrical structure is a latent subspace, i.e., a linear generative model. Specifically, models such as PPCA consider the neural activity as a linear map $\bm{W}\in \reals^{N \times P}$ applied to a $P$-dimensional latent variable $\bm{z}_t \in \reals^P$, where $P \ll N$. By further modeling the observation error as Gaussian and the latent state as having a Gaussian prior, PPCA results in the model where $\bm{W}$ is the latent-to-neural-activity matrix, and $\bm{\Phi}$ is the diagonal observation noise covariance. Extensions to this base model have included temporal smoothness priors over $\bm{z}_t$ via Gaussian processes~\cite{yu2008a, keeley2020identifying}. In these cases the prior over $\bm{z}_t$ is not factorizable over time $t$, resulting in an effective matrix-normal prior $\bm{Z} = [\bm{z}_1,\dots,\bm{z}_T] \sim \mathcal{MN}(\bm{0}, \bm{I}, \bm{K})$, where $\bm{K}$ is a GP covariance matrix with symmetric kernel $k(\cdot,\cdot)$: $\bm{K}_{ij} = k(t_i, t_j)$. These GP-based latent variable models (GP-LVMs) can be adapted by varying $k(\cdot,\cdot)$ to adjust the temporal smoothness~\cite{lawrence2005, yu2008a, santhanam2009,wu2017}. In these cases, observations are modeled as Gaussian as opposed to Poisson.

\paragraph{Coupled spike-generation model:} Another class of models characterizes inter-neuron dependencies more explicitly through the spike-generation mechanism. In the generalized linear population model of~\cite{pillow2008}, each neuron's conditional intensity is driven by a linear combination of the stimulus, its own spike history, and filtered spike trains from the other neurons in the population,

\begin{equation}\label{eq:glm_coupled}
	\lambda_i(t) = g\!\left(\bm{k}_i^\top \bm{x}(t) + \sum_{j=1}^{N} (\bm{h}_{ij} * y_j)(t) + b_i\right),
\end{equation}

where $\bm{k}_i \in \reals^d$ is the stimulus filter for neuron $i$, $\bm{h}_{ij}$ is a temporal coupling filter from neuron $j$ to neuron $i$, $*$ denotes convolution with the spike train $y_j$, $b_i$ is a scalar bias, and $g(\cdot)$ is a point nonlinearity. The GLM-style formulation captures pairwise spiking dependencies and can reproduce correlated firing, but there is no latent variable that co-modulates firing rates across the network between spikes.


\section{Poisson matrix normal latent variable model}\label{sec:pmnlv}
We first specify the generative model, then derive two complementary estimation algorithms in Sections~\ref{sec:vem} and~\ref{sec:ktm}.

\begin{figure}[!t]
	\centering
	\includegraphics[width=\linewidth]{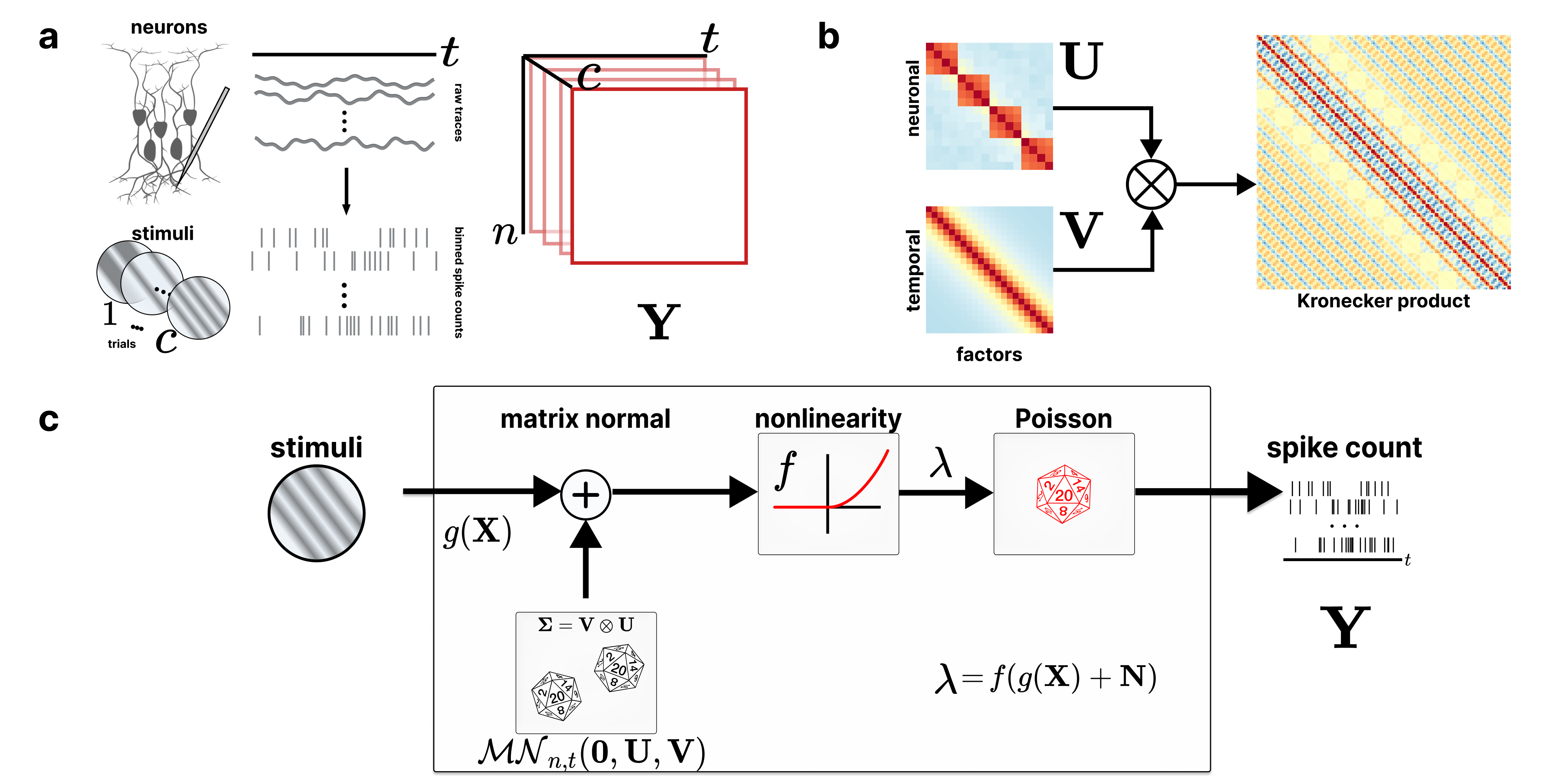}
	\caption{Poisson matrix-normal latent variable model. (a) Data collection from $N$ simultaneously recorded neurons, $C$ experimental conditions, and $T$ time steps yields a data tensor $\mathcal{Y}$. (b) A block-structured $NT \times NT$ Kronecker covariance matrix (right) decomposes into a neuron factor $\bm{U}$ (top left) capturing inter-neuron correlations and a time factor $\bm{V}$ (bottom left) capturing temporal covariance, which serves as the prior for the generative model. (c) The stimulus-driven tuning term $f(\bm{X})$ and a zero-mean matrix-normal gain $\mathcal{N} \sim \mathcal{MN}(\bm{0}, \bm{U}, \bm{V})$ are summed and passed through the quadratic soft-rectifying link $g(\cdot)$ before Poisson spike generation.}
\label{fig:schema}
\vspace{-0.5cm}
\end{figure}

\paragraph{Generative model:}\label{sec:gen_model}
We consider simultaneously measured neuronal electrophysiological recordings from a population of neurons $n\in\{1,\dots, N\}$ with binned times $t\in\{1,\dots, T\}$ across multiple trials or conditions $c\in\{1,\dots, C\}$. Each condition indexes an observed data matrix $\bm{Y}_c \in \mathbb{R}^{N \times T}$, which stacked across conditions produces a tensor $\mathcal{Y}$ of dimensions $N\times T\times C$. Following the single-neuron framework~\cite{charles2018}, we model each count as Poisson with a rate determined by the sum of a stimulus-driven component and a latent population gain,
\begin{equation}\label{eq:obs}
	\bm{Y}_c \mid \bm{N}_c \sim \mathrm{Poisson}\!\left(g\!\left(f(\bm{X})_c + \bm{N}_c\right)\right),
\end{equation}
where $f(\bm{X})_c \in \mathbb{R}^{N \times T}$ is the matrix of stimulus tuning values for condition $c$ and $g(\cdot) = \mathrm{softplus}(\cdot)^2$ is a quadratic soft-rectifying link applied elementwise which was chosen so that the marginal variance grows similarly to the overdispersion of recorded cortical data. The key departure from the single-neuron model is that the latent gain $\bm{N}_c \in \mathbb{R}^{N \times T}$ is no longer independent across neurons; its prior structure encodes the population co-variability across neurons and time. 
In extending the single-neuron gain prior to a population, we model each $\bm{N}_c$ as a matrix-normal random matrix: $\bm{N}_c\sim\mathcal{MN}(\bm{0},\bm{U}, \bm{V})$,
separating neuronal and temporal structure into interpretable factors. 
Here $\bm{N}_c$ is drawn independently across conditions, treating $c$ as the replicate axis. The same prior applies when $c$ indexes trials, where the row factor $\bm{U} \in \Spp{N}$ encodes inter-neuron gain correlations and the column factor $\bm{V} \in \Spp{T}$ encodes temporal covariance. While $\bm{\Sigma}$ is identified, with full covariance $\bm{\Sigma} = \bm{V} \otimes \bm{U} \in \reals^{NT \times NT}$, its Kronecker factors are identifiable up to a positive scalar~\cite{werner2007}. Each of our estimation algorithms resolves this ambiguity by its own normalization convention, described below.

\paragraph{Marginal likelihood: }\label{sec:marginal}
We estimate the model parameters by maximizing the marginal likelihood over $\Theta = \{f(\bm{X}), \bm{U}, \bm{V}\}$,
\begin{equation}\label{eq:marginal}
	\operatorname{p}(\bm{Y} \mid \Theta) = \int \operatorname{p}(\bm{Y} \mid \bm{N}, \Theta)\, \operatorname{p}(\bm{N} \mid \bm{U}, \bm{V})\;d\bm{N}.
\end{equation}
The nonlinear link $g(\cdot)$ and the non-conjugacy between the Poisson likelihood and the matrix-normal prior, however, render this integral intractable. Integrals from this model type are often computed using the Laplace approximation~\cite{rasmussen2008, cunningham2007, rue2009}, but this approach loses approximation accuracy in high-dimensions~\cite{shun1995, tierney1989}. We therefore propose two complementary estimation algorithms. The first (Section~\ref{sec:vem}) places no structural assumptions on $\bm{U}$ or $\bm{V}$ and recovers dense factor estimates via variational EM (VEM). The second restricts the Kronecker factors to a parametric kernel family, trading generality for scalability via composite likelihood (Section~\ref{sec:ktm}).

\section{Variational Expectation-Maximization (VEM)}\label{sec:vem}
As direct marginalization over $\bm{N}$ in~\eqref{eq:marginal} is intractable, we instead maximize a variational lower bound on $\log\operatorname{p}(\bm{Y} \mid \Theta)$~\cite{blei2017},
\begin{align}\label{eq:elbo}
	\mathcal{L}(\bm{M}, \bm{\Psi}_U, \bm{\Psi}_V) &= \mathbb{E}_{\operatorname{q}}\!\left[\log \operatorname{p}(\bm{Y} \mid \bm{N}, \bm{X})\right] - \operatorname{KL}\!\left(\operatorname{q}_\Psi(\bm{N}) \| \mathcal{MN}(\bm{0}, \bm{U}, \bm{V})\right) \notag \\
	&= \mathbb{E}_{\operatorname{q}}\!\left[\log \operatorname{Poisson}\!\left(\bm{Y} \mid g\!\left(f(\bm{X}) + \bm{N}\right)\right)\right] \\
    &- \operatorname{KL}\!\left(\mathcal{MN}(\bm{M}, \bm{\Psi}_U, \bm{\Psi}_V) \| \mathcal{MN}(\bm{0}, \bm{U}, \bm{V})\right). \notag
\end{align}
A standard mean-field approximation $\operatorname{q} = \prod_{n,t}q_{nt}$ factorizes across all $(n,t)$, collapsing $\bm{\Sigma}$ to a diagonal in the posterior and biasing the M-step toward the identity. Instead, we select a matrix-normal variational posterior and prevent collapse by the addition of an accumulating SCM in both $\bm{\Psi}_U$ and $\bm{\Psi}_V$ during the EM procedure~\cite{chiquet2019}. We compute the matrix-normal variational posterior,
$\operatorname{q}_\Psi(\bm{N}) = \mathcal{MN}_{N \times T}(\bm{M},\, \bm{\Psi}_U,\, \bm{\Psi}_V)$,
with variational mean $\bm{M} \in \mathbb{R}^{N \times T}$, Kronecker factors $\bm{\Psi}_U \in \Spp{N}$, $\bm{\Psi}_V \in \Spp{T}$, and $\mathbb{E}_{\operatorname{q}}$ denoting expectation under $\operatorname{q}_\Psi(\bm{N})$. We parameterize each factor by its Cholesky factor, reducing the per-iteration M-step cost to $\mathcal{O}(N^3+T^3)$ and constrain optimization to $\Spp{N}$ and $\Spp{T}$~\cite{boyd_convex_2004}.

\paragraph{E-step:} We first hold $\Theta$ constant and optimize over the variational parameters $(\bm{M}, \bm{\Psi}_U, \bm{\Psi}_V)$ via gradient ascent. Since the Poisson likelihood factorizes over each $(n,t,c)$ entry and only the marginal variance enters for each observed count (see Supplement~\ref{sec:supp_ghq}), the expected log-likelihood decomposes into $NTC$ independent 1-D integrals~\cite{aitchison1989a}, each approximated by $K$-point Gauss-Hermite quadrature (GHQ) with nodes and weights $\{z_k, w_k\}_{k=1}^K$.
\paragraph{M-step:} Following the E-step, $\Theta$ is updated via the flip-flop procedure~\cite{bien2011}, holding each Kronecker factor fixed while updating the other, until the maximum relative element-wise change in both $\widehat{\bm{U}}$ and $\widehat{\bm{V}}$ plateaus. Let $\bm{M}_c \in \mathbb{R}^{N \times T}$ denote the $c^{th}$ condition-slice of the posterior mean $\bm{M}$. Maximizing the ELBO with respect to $\bm{U}$ holding $\bm{V}$, $\bm{\Psi}_U$, $\bm{\Psi}_V$ fixed yields
\begin{equation}\label{eq:Uhat}
	\widehat{\bm{U}} = \frac{1}{TC}\sum_{c=1}^{C} \bm{M}_c\bm{V}^{-1}\bm{M}_c^\top + \frac{\operatorname{tr}(\bm{V}^{-1}\bm{\Psi}_V)}{C}\,\bm{\Psi}_U,
\end{equation}
with a symmetric update for $\widehat{\bm{V}}$ using the updated $\widehat{\bm{U}}^{-1}$ (full derivation in Supplement ~\ref{sec:supp_mstep}). The first term is a weighted sample covariance of the posterior means. The second propagates residual posterior uncertainty into the factor estimate, preventing shrinkage of $\widehat{\bm{U}}$ across EM iterations. After convergence, $\operatorname{tr}(\widehat{\bm{U}})$ is rescaled to $N$, resolving the scalar ambiguity of the Kronecker product. The Kronecker factorization reduces the per-iteration cost from $\mathcal{O}(N^3T^3)$ for direct inversion of $\bm{\Sigma} = \bm{V} \otimes \bm{U} \in \mathbb{R}^{NT \times NT}$ to $\mathcal{O}(N^3 + T^3)$ via independent Cholesky factorizations of $\bm{U}$ and $\bm{V}$, exploiting the mixed-product identity $(\bm{V} \otimes \bm{U})^{-1} = \bm{V}^{-1} \otimes \bm{U}^{-1}$.

\section{Kernel Tournament Method (KTM)}\label{sec:ktm}
In addition to the VEM fitting approach we, propose a complementary parametric method which estimates $\Theta$ by exploiting two observations: (i) tuning curves estimation can be reformulated as a composite likelihood optimization over pairs of observations rather than the full joint and (ii) domain knowledge of the experiment or stimulation can constrain the prior covariance structure. As the marginalized joint likelihood~\eqref{eq:marginal} is intractable, we can instead well approximate the likelihood by direct approximation of the integral~\cite{aitchison1989a} using GHQ (see Supplement~\ref{sec:supp_ghq2}). Experimental knowledge reformulates the estimation of $\bm{U}$ and $\bm{V}$ into an estimation over a small set of hyperparameters.

\paragraph{Composite likelihood:} Composite likelihood (CL) estimation replaces the intractable full joint $p(\mathbf{Y}\mid\Theta)$ with a product of lower-dimensional marginals. Because each marginal term is itself a valid likelihood, the CL estimator is consistent under standard regularity conditions~\cite{varin2011}. Bivariate marginals encode the dominant covariance structure of $\bm{U}$ and $\bm{V}$ while reducing each intractable $NT$-dimensional integral to a two-dimensional integral tractable by 2D GHQ (see Supplement~\ref{sec:supp_ghq2}). Let $\mathcal{S} = \{(a,b)\}$ be the set of index pairs with $a=(n,t)$ and $b=(n',t')$. The pairwise composite log-likelihood is:\begin{equation}\label{eq:cl}
	\operatorname{CL}(\Theta;\,\bm{Y},\,\mathcal{S})
	= \sum_{(a,\,b)\,\in\,\mathcal{S}} \log p(y_a,\,y_b \mid \Theta),
\end{equation}
where each bivariate marginal $p(y_a, y_b \mid \Theta)$ integrates the latent gain pair $(z_a, z_b)$ over the bivariate margin of $\mathcal{MN}(\mathbf{0}, \mathbf{U}, \mathbf{V})$ with additional details in Supplement \ref{sec:supp_ghq_conv}.

\paragraph{Stage 1 - tuning curve estimation:}
The tuning curves $\widehat{f} = \widehat{f(\bm{X})}$ are estimated by minimizing the negative $\operatorname{CL}(\Theta;\,\bm{Y},\,\mathcal{S})$ over only $f$, with $\bm{U}$ and $\bm{V}$ held at their initializations. 
Each bivariate marginal $p(y_a, y_b \mid \Theta)$ in Equation~\eqref{eq:cl} takes the same Poisson-mixture form as the marginal likelihood in Equation~\eqref{eq:marginal}, restricted to the pair $(a, b)$ and evaluated by 2D GHQ (see Supplement~\ref{sec:sup_KTM_1} for additional details).

\paragraph{Stage 2 - kernel family selection}
Having computed $\widehat{f}$, variance-stabilized residuals are aggregated into a neural-axis SCM $\bm{S}_U$ and time-axis SCM $\bm{S}_V$ by averaging over columns and rows respectively. Each axis is routed independently to either a low-rank (LR) or full-rank (FR) structural family. For LR, low-rankness is detected via the method in~\cite{saibaba2023} and the spectral gap. The detected rank seeds a truncated eigendecomposition as the candidate matrix. For FR, the kernel dictionary consists of Mat\'{e}rn ($\nu\in\{1.5,2.5,\infty\}$), with $\infty$ the squared-exponential and periodic families, selected by minimizing the masked Frobenius loss over the off-diagonal entries of the residual SCM,\begin{equation}\label{eq:wishart_pl}
	\mathcal{L}_{W}(\bm{K};\,\bm{S})
	= \left\|\bm{D}\odot\bigl(\bm{S} - \hat{a}\bm{K}\bigr)\right\|_{F}^{2},
	\quad
	\hat{a} = \frac{\langle \bm{D}\odot\bm{S},\,\bm{D}\odot\bm{K}\rangle}
	{\|\bm{D}\odot\bm{K}\|_{F}^{2}},
\end{equation}
where $\bm{D} = \bm{1}\bm{1}^{\top} - \bm{I}$ masks the diagonal to measure the off-diagonal features. A periodicity gate based on the lag-averaged correlation function suppresses periodic fits when spectral evidence is insufficient. Each candidate is scouted over a fixed number of iterations minimizing Equation~\eqref{eq:cl} on $\mathcal{S}$ with free kernel hyperparameters (see Supplement~\ref{sec:sup_KTM_2} for additional details).

\paragraph{Stage 3 - joint refinement}
The winning family $k^{*}$ and its scouted factors $\bm{U}_{k^{*}}$ and $\bm{V}_{k^{*}}$ are passed to a joint refinement of $(\hat{f},\,\bm{U}_{k^{*}},\,\bm{V}_{k^{*}})$ by minimizing $-\operatorname{CL}(\Theta;\,\bm{Y},\,\mathcal{S})$ over all parameters simultaneously until convergence (see Supplement~\ref{sec:sup_KTM_3} for additional details).

\section{Results on Simulated data}\label{sec:sim}
We first assessed recovery of the Kronecker factors on simulated data sampled from the generative model with dimensions $N{=}45$, $T{=}12$, $C{=}75$ (matched to real data). Observed data is first variance stabilized by the standard square-root transform~\cite{yu2008a, bartlett_square_1936}. We assessed recovery by relative Frobenius error $\varepsilon_F(\hat{\bm{U}})=\fronorm{\widehat{\bm{U}} - \bm{U}} /\fronorm{\bm{U}}$ as the primary metric and off-diagonal correlation $\rho(\widehat{\bm{U}})$ as a secondary metric of structural fidelity with hardware details in Supplement ~\ref{sec:supp_machine}.

We evaluated VEM and KTM across the four covariance regimes formed by independently assigning $\bm{U}$ and $\bm{V}$ to either a low-rank (LR) or Mat\'{e}rn (M) ground truth. Because cortical populations are often low-dimensional, we focus on the regime where $\bm{U}$ is LR and $\bm{V}$ is M (for detailed results see Supplement~\ref{sec:supp_proposebench}). 

 \begin{figure}[!t]
     \centering
     \includegraphics[width=\linewidth]{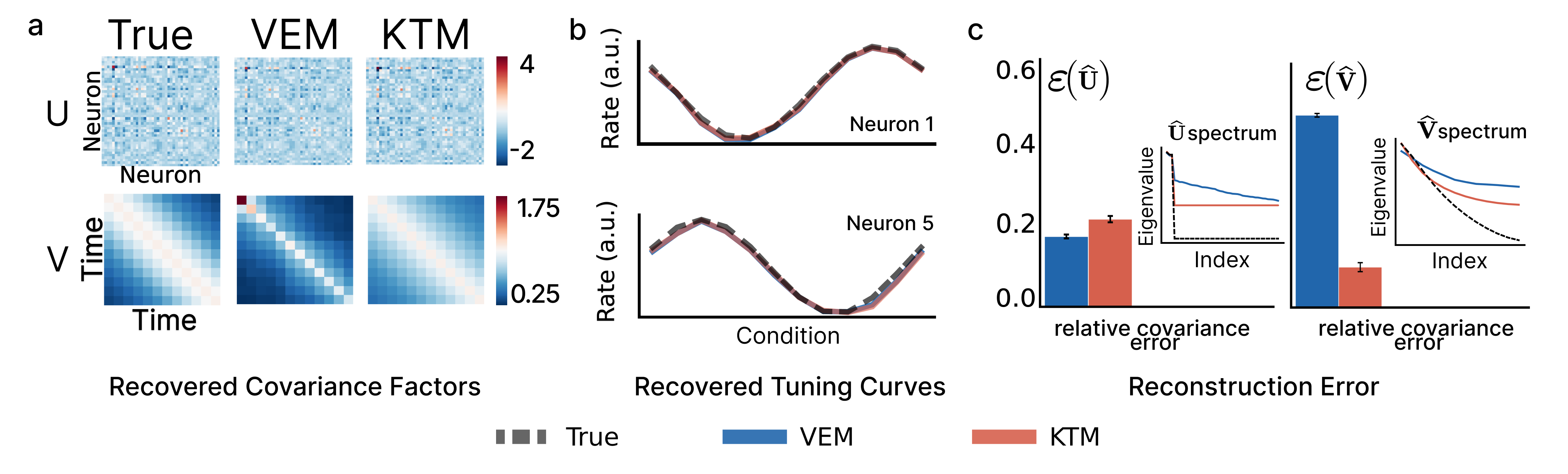}
     \caption{Covariance and tuning curve recovery in the LR$\times$M regime. (a) Heatmaps of ground truth, VEM, and KTM estimates of the neural factor $\bm{U}$ (top) and temporal factor $\bm{V}$ (bottom) with inset eigenspectra (black: true, blue: VEM, red: KTM). b) Estimated tuning curves $\widehat{f}(\bm{X})$ for four example neurons in no order. (c) Relative Frobenius error $\varepsilon_F(\widehat{\bm{U}})$ (top) and $\varepsilon_F(\widehat{\bm{V}})$ (bottom) averaged over $n=12$ seeds (error bars: $\pm$ SEM). VEM achieves lower error on the lower-rank neural factor; KTM achieves lower error on the stationary temporal factor.}
     \label{fig:sim}
     \vspace{-0.5cm}
 \end{figure}

Both methods recovered the tuning curves to a similar degree of accuracy ($\varepsilon_F(\widehat{f}) \approx 10\%$), indicating that both algorithms can accurately capture the neural population's stimulus encoding while also providing estimates of the covariance structure. The two methods diverged in their recovery of the individual Kronecker factors consistent with the Kronecker non-identifiability (Section~\ref{sec:background}). VEM, which places no parametric constraint on the factors, recovered $\bm{U}$ more accurately than KTM ($17.0 \pm 0.5\%$ vs.\ $21.3 \pm 0.8\%$) and off-diagonal correlation $\rho\left(\widehat{\bm{U}}\right){=}99.1\%\pm0.1\%$. We interpret this to  suggest that  the unconstrained matrix-normal variational posterior is better suited to the low-rank geometry of $\bm{U}$. Conversely, KTM recovered the temporal factor $\bm{V}$ more accurately than VEM ($9.9 \pm 1.1\%$ vs.\ $47.3 \pm 0.4\%$). This outcome is consistent with the Kronecker non-identifiability as KTM optimizes the kernel directly while VEM jointly optimizes the latent covariance matrix and firing rate (Table~\ref{tab:compare}). This pattern is further supported by KTM's identified kernel spectral decay closely tracking the ground truth (Fig.~\ref{fig:sim}, inset), while VEM approximates the spectral structure well.

\begin{table}[!t]
\caption{Comparison with $N{=}45$, $C{=}12$, $T{=}75$ between means over 12 seeds. Best per row in bold}
\label{tab:compare}
\centering
\small
\begin{tabular}{llcccccc}
\textbf{prior structure}& & \multicolumn{6}{c}{\textbf{method mean relative error}} \\
& & VEM & KTM & GPLVM & Goris & GPFA & MLE (GH) \\
\cline{3-8}
\noalign{\vskip 1.5pt}
\cline{3-8}
  & $\varepsilon(\widehat{\mathbf{U}})$ & \textbf{0.173} & 0.213 & 0.966 & 0.360 & 1.296 & 0.933 \\
 \textit{LR-Mat} & $\varepsilon(\widehat{\mathbf{V}})$ & 0.479 & \textbf{0.069} & 0.834 & 0.451 & 0.483 & 0.908 \\
  & $\varepsilon(\hat{f})$ & 0.127 & \textbf{0.098} & 0.175 & 0.470 & 0.470 & 0.107 \\
\cline{3-8}
\end{tabular}
\end{table}

\section{Results on visual coding dataset}\label{sec:real}

 \begin{figure}[!t]
     \centering
     \includegraphics[width=\linewidth]{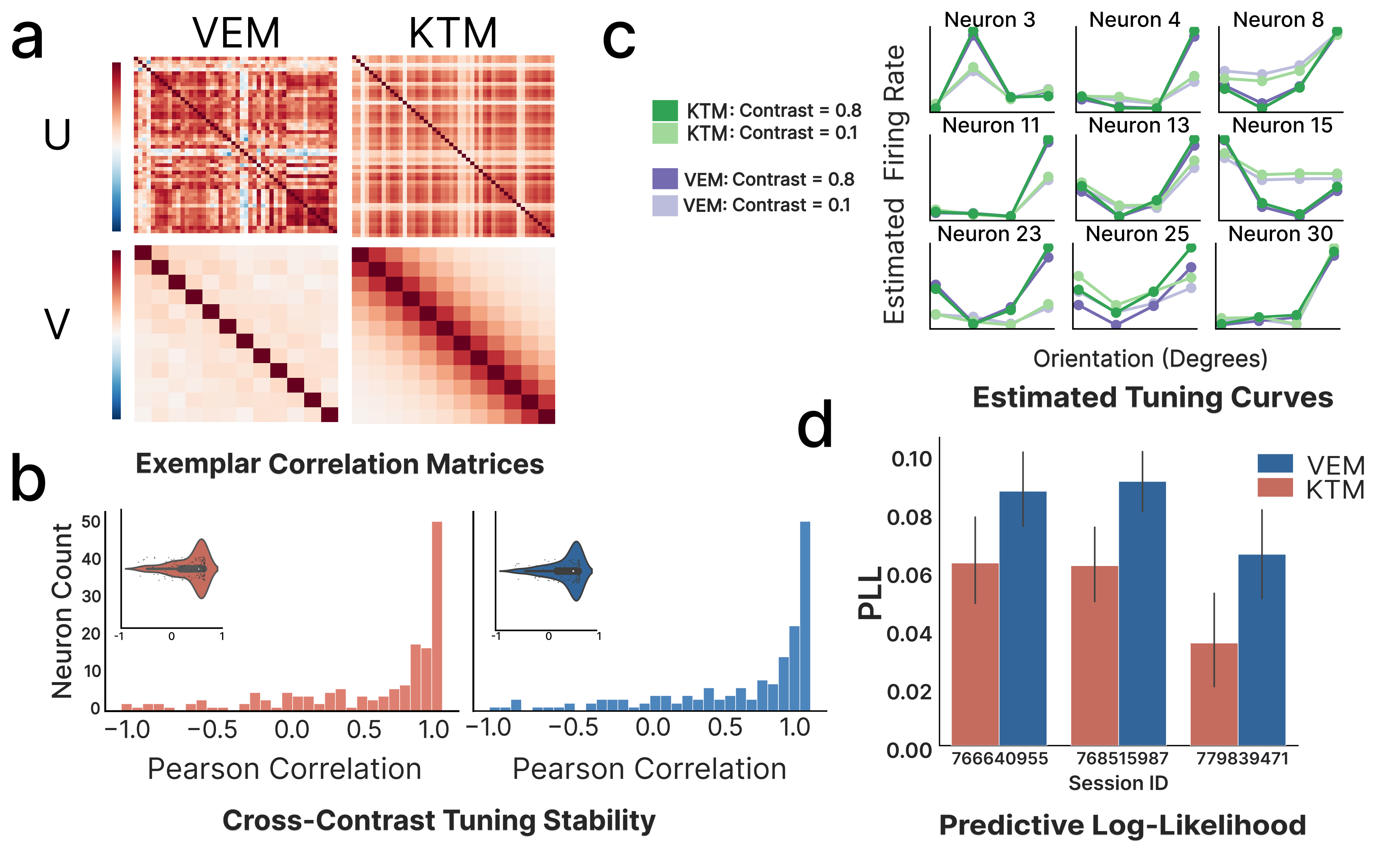}
     \caption{VEM and KTM applied to the Allen Institute Visual Coding dataset from VISp at contrast $0.8$ and $90^\circ$ orientation. (a) Estimated latent covariance matrices normalized to show correlation of neuronal and temporal dimensions. (b) histogram of cross-contrast tuning stability in VISp. (c) Example estimated tuning curves. (d) PLL of VEM and KTM.}
     \label{fig:visp}
     \vspace{-0.5cm}
 \end{figure}

To demonstrate the utility of PMNLV for studying neural population dynamics, we next applied VEM to Neuropixels recordings from the Allen Institute Visual Coding project. This dataset contains simultaneous recordings of single unit activity from distinct regions of mouse visual processing~\cite{siegle_survey_2021}. During recording, mice were head-fixed and free to run on a wheel while visual stimuli were presented. We focus on four consecutive regions along the visual hierarchy: dorsal lateral geniculate nucleus (LGd), primary visual area (VISp), rostrolateral visual area (VISrl), and anteromedial visual area (VISam). We isolate recordings taken during the presentation of drifting grating stimuli of 4 orientations (0$^{\circ}$, 45$^{\circ}$, 90$^{\circ}$, 135$^{\circ}$) and 2 contrast levels (0.1, 0.8) presented for 2~s and repeated 75 times. To account for differences in visual processing latency, we excluded the first 200~ms of stimulus presentation and binned the subsequent 1.8~s of spike counts into 150~ms windows. We strictly included units that passed standardized spike sorting quality control metrics (ISI violations $< .5$, presence ratio $>.95$, amplitude cutoff $<.1$) and had mean firing rate $>2.0$~Hz. We selected the three sessions with the most isolated units that had $\geq20$ units per brain region.

We first validated PMNLV by evaluating its held-out predictive log-likelihood (PLL) against a baseline mean-firing rate model with a 70-30 data split for each session/stimulus (Fig.~\ref{fig:visp}).
We then subtracted the log-likelihood of the baseline PMNLV model, parameterized by the empirical mean firing rate of each neuron in each time-bin, from our model’s log-likelihood. We found that VEM and KTM consistently outperform the baseline by $.067 - .090$ and $.036 - .065$ nats per observation, respectively (Fig.~\ref{fig:visp}d). Moreover, both VEM and KTM recover plausible orientation-selective tuning curves whose baseline preferences remain stable across contrast conditions (Fig.~\ref{fig:visp}c). Together, these results indicate that our model is both generalizable and captures realistic orientation selective tuning curves in VISp phenomena~\cite{niell_highly_2008}.

 \begin{figure}[!t]
     \centering
     \includegraphics[width=\linewidth]{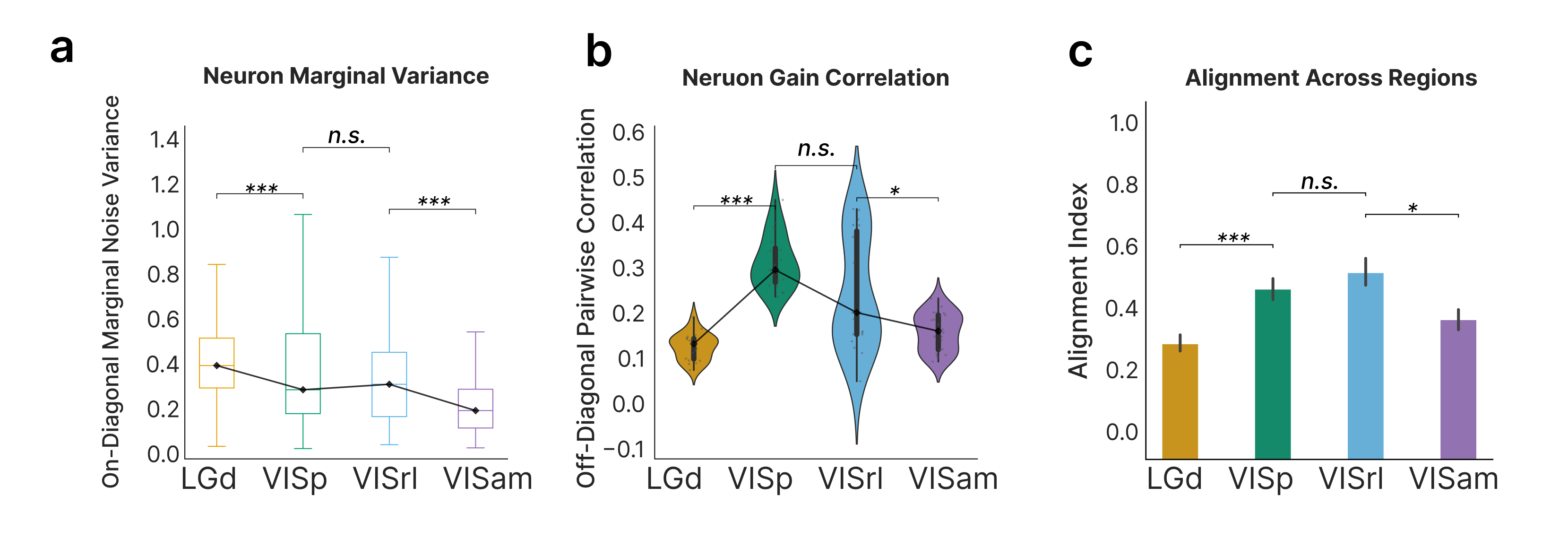}
     \caption{VEM applied to the Allen Institute Visual Coding dataset across four connected brain regions. (a) Marginal variance across regions with no significant difference between VISp and VISrl. (b) Off-diagonal correlation significant differences from LGd to VISp and VISrl to VISam (c) Alignment between neural activity and $\hat{\bm{U}}$ ***, p < 0.001, *, p < 0.05, computed by Wilcoxon rank-sum test.}
     \label{fig:hierarchy}
     \vspace{-0.5cm}
 \end{figure}

We next investigated the question: how does population variability change across the visual hierarchy? Consistent with past work~\cite{akella_deciphering_2025}, we found that the marginal variances of our model, which represent single-neuron variability, change minimally along the hierarchy (Fig.~\ref{fig:hierarchy}a). However, we are also uniquely positioned to explore how population covariance evolves along the hierarchy. To do so, we  calculated the mean absolute value across the off-diagonal terms of the neuron gain correlation matrices (Fig.~\ref{fig:hierarchy}b). We find that, in contrast to single-neuron variability, shared co-variability is highly region dependent, with gain correlation initially low in the thalamus (LGd), peaking in VISp, and then systematically decreasing into higher-order visual areas (VISrl, VISam). This pattern suggests that while individual neurons remain similarly variable, ascending the cortical hierarchy decorrelates the population level co-variability~\cite{siegle_survey_2021, jia_modules_2022, montjin_populationcodes_2016}. 

Finally, we explored the difference between co-variability in neural activity and neural gain (see Supplement~\ref{sec:supp_subspace}). Specifically, we computed the alignment index~\cite{elsayed_reorganization_2016} between the covariance of the neural activity and the first 10 principal components of the recovered neuron gain covariance matrices. As with the total neural gain co-variability, alignment begins low in LGd, peaks in early stages of visual cortex (VISp and VISrl), and then declines in higher order cortex (VISam). We note that, generally, the subspaces defining the neural activity and neuron noise covariance are substantially misaligned (alignment index $<0.6$). This discrepancy indicates that the drivers of neural gain covariance may operate separately from the activity covariance. 

\section{Conclusion}\label{sec:conclusion}
We introduce here the Poisson matrix-normal latent variable (PMNLV) model for population-level overdispersion. PMNLV is a generative framework that simultaneously respects the observed Poisson spiking statistics and the correlated gain fluctuations across neurons and time. We derive two complementary estimation algorithms: (1) VEM, which recovers dense Kronecker factors without structural assumptions and (2) KTM, a data-driven method that constrains the prior to a parametric family. We show on synthetic data that both the VEM and KTM outperform standard methods in recovering the Kronecker factors. 

By applying our methods to simulated data reflective of empirical neural recordings, we demonstrate that both methods can accurately recover ground truth covariance matrices for the matrix normal prior. We found that when structural priors are clear, the KTM method produces an improved covariance estimate over VEM. Otherwise, when no structural priors can be assumed, VEM recovers both the tuning curves and covariance matrices of the population. Both methods only require that the data be simultaneously recorded across neurons and time-bins without assumptions on the stimulus design. 

In applying the PMNLV model to mouse visual data, we demonstrate that our proposed methods capture realistic single-neuron tuning curves as well as exhibit generalizability through PLL validation. We also show, through the estimated correlation, that the within brain region variance tends to show little change between VISp and VISrl regions and decreases across the hierarchy. Although the change between VISp and VISrl is not significant, we found that the noise co-variability in spiking data surprisingly increases after LGd. Finally, by applying these methods to study the geometry of gain covariance, we also show that, unexpectedly, the noise covariance is misaligned with the covariance of the neural activity. In sum, the PMNLV model extends findings from studies on single-neuron variability to a broader context. Future work will extend both methods to quantify uncertainty over the recovered Kronecker factors. 

\paragraph{Limitations:}\label{sec:limitations}
The Kronecker factorization $\bm{\Sigma} = \bm{V}\otimes\bm{U}$ assumes separability of neuronal and temporal covariance and so neuron-specific temporal dynamics are not captured. However, taking time-bins small enough, the neural activity can be taken to be approximately stationary which is the convention for many neuroscience applications~\cite{aghamohammadi2025, brown_multiple_2004}. Nevertheless, reliable factor recovery relies on sufficient observations across each observed dimension to provide good sample covariance estimates. Both VEM and KTM operate in this regime. For VEM specifically, the Cholesky factorization reduces the computational cost per M-step iterate but on a CPU this can become a bottleneck for large $N$ or large $T$ which can be resolved using a GPU or use of the KTM method. 

\paragraph{Broader Impact:}\label{sec:broaderimpact}
This work introduces a basic science model for characterizing population-level gain covariance in simultaneously recorded neural populations. We envision the primary beneficiaries as experimental neuroscientists seeking to quantify how shared variability is organized across the brain.
We do not foresee any direct negative societal impacts as the methods operate on anonymized electrophysiological recordings from non-human animal subjects and involve no personal data or decision-making systems. The main indirect risk arises if the model is applied to human neural recordings in a clinical context, where mischaracterized covariance could support incorrect inferences about neural state. 

Ethical aspects and future societal consequences do not apply to this work.

\begin{ack}
This work was supported by NIH grants R01DC002390, R01DC018061, T32DC000023, T32GM136577, and NSF CAREER Award 2340338. 
\end{ack}

\printbibliography

\newpage
\cleardoublepage
\part*{Supplement}
\addcontentsline{toc}{part}{Appendix}

\appendix

\renewcommand{\theequation}{A\arabic{equation}}
\renewcommand{\thefigure}{A\arabic{figure}}
\renewcommand{\thetable}{A\arabic{table}}
\setcounter{equation}{0}
\setcounter{figure}{0}
\setcounter{table}{0}

\section{Gauss--Hermite Quadrature}\label{sec:supp_ghq}

The nonlinear link $g(\cdot) = \mathrm{softplus}(\cdot)^2$ makes the expected log-likelihood and the bivariate marginals in~\eqref{eq:cl} analytically intractable. We approximate both with Gauss--Hermite quadrature (GHQ)~\cite{abramowitz1965}. For a function $f$ weighted against a Gaussian kernel, the $K$-point rule is
\begin{equation}\label{eq:ghq_generic}
  \frac{1}{\sqrt{\pi}}\int_{-\infty}^{\infty} f(\xi)\,e^{-\xi^2}\,d\xi
  \;\approx\;
  \frac{1}{\sqrt{\pi}}\sum_{k=1}^{K} w_k\,f(\xi_k),
\end{equation}
where $\{\xi_k,w_k\}_{k=1}^K$ are the standard Hermite nodes and weights.

\subsection{1-D quadrature (VEM)}\label{sec:supp_ghq1}

Each entry of the expected log-likelihood under $q_\Psi$ is a scalar Gaussian integral. The substitution $z = \mu_z + \sqrt{2}\,\sigma\,\xi$ with $\mu_z = [f(\bm{X})]_{rc} + [\bm{M}]_{rc}$ and $\sigma^2 = [\bm{\Psi}_U]_{rr}[\bm{\Psi}_V]_{cc}$ reduces it to~\eqref{eq:ghq_generic}. The full expected log-likelihood sums these over all $(r,c)$ entries and conditions.

\subsection{2-D quadrature (KTM)}\label{sec:supp_ghq2}

Each bivariate marginal $p(y_a,y_b \mid \Theta)$ in~\eqref{eq:cl} integrates the Poisson likelihood against a bivariate Gaussian over the latent pair $(z_a,z_b)$ with $a=(n,t)$, $b=(n',t')$. The $2\times 2$ marginal covariance is
\begin{equation}\label{eq:biv_marginal}
  \bm{\Sigma}_{ab}
  = \begin{pmatrix}
      [\bm{U}]_{nn}[\bm{V}]_{tt}     & [\bm{U}]_{nn'}[\bm{V}]_{tt'} \\[3pt]
      [\bm{U}]_{nn'}[\bm{V}]_{tt'}   & [\bm{U}]_{n'n'}[\bm{V}]_{t't'}
    \end{pmatrix},
\end{equation}
following from $\bm{\Sigma} = \bm{V}\otimes\bm{U}$. Let $\bm{L}$ be the Cholesky factor of $\bm{\Sigma}_{ab}$. The substitution $\bm{z} = \bm{\mu}_{ab} + \sqrt{2}\,\bm{L}\,\bm{\xi}$ absorbs the Gaussian density into $e^{-\|\bm{\xi}\|^2} = e^{-\xi_1^2}e^{-\xi_2^2}$, giving the tensor-product rule
\begin{equation}\label{eq:2d_ghq}
  p(y_a,y_b \mid \Theta)
  \;\approx\;
  \frac{1}{\pi}
  \sum_{k_1=1}^{K}\sum_{k_2=1}^{K}
  w_{k_1}\,w_{k_2}\;
  \prod_{j \in \{a,b\}}
  \frac{g(\tilde{z}_j)^{y_j}\,e^{-g(\tilde{z}_j)}}{y_j!},
\end{equation}
where $\tilde{\bm{z}} = \bm{\mu}_{ab} + \sqrt{2}\,\bm{L}(\xi_{k_1},\xi_{k_2})^\top$.

\subsection{Convergence of the quadrature approximation}\label{sec:supp_ghq_conv}

To avoid the combinatorial hurdle of enumerating all neuron pairs $(n,n')$, pairs are subsampled from the full tensor slice. Sampling probability was is weighted by $\ell_s$ and decays in proportion to $\exp(-d_{nn'}/\ell_s)$, concentrating pairs on units with stronger local covariance structure. Similarly, time pairs are weight by $\tau$ and restricted to a window of $|t-t'|\leq\tau$. We verify that the quadrature approximations converge at low node counts. For a representative synthetic dataset ($N=45$, $T=60$, $C=75$), we computed the 1D expected log-likelihood and the 2D bivariate marginal~\eqref{eq:2d_ghq} at $K = 2,3,\ldots,20$ and measured the relative error against a reference at $K=30$.

\begin{figure}[H]
  \centering
  \includegraphics[width=0.55\linewidth]{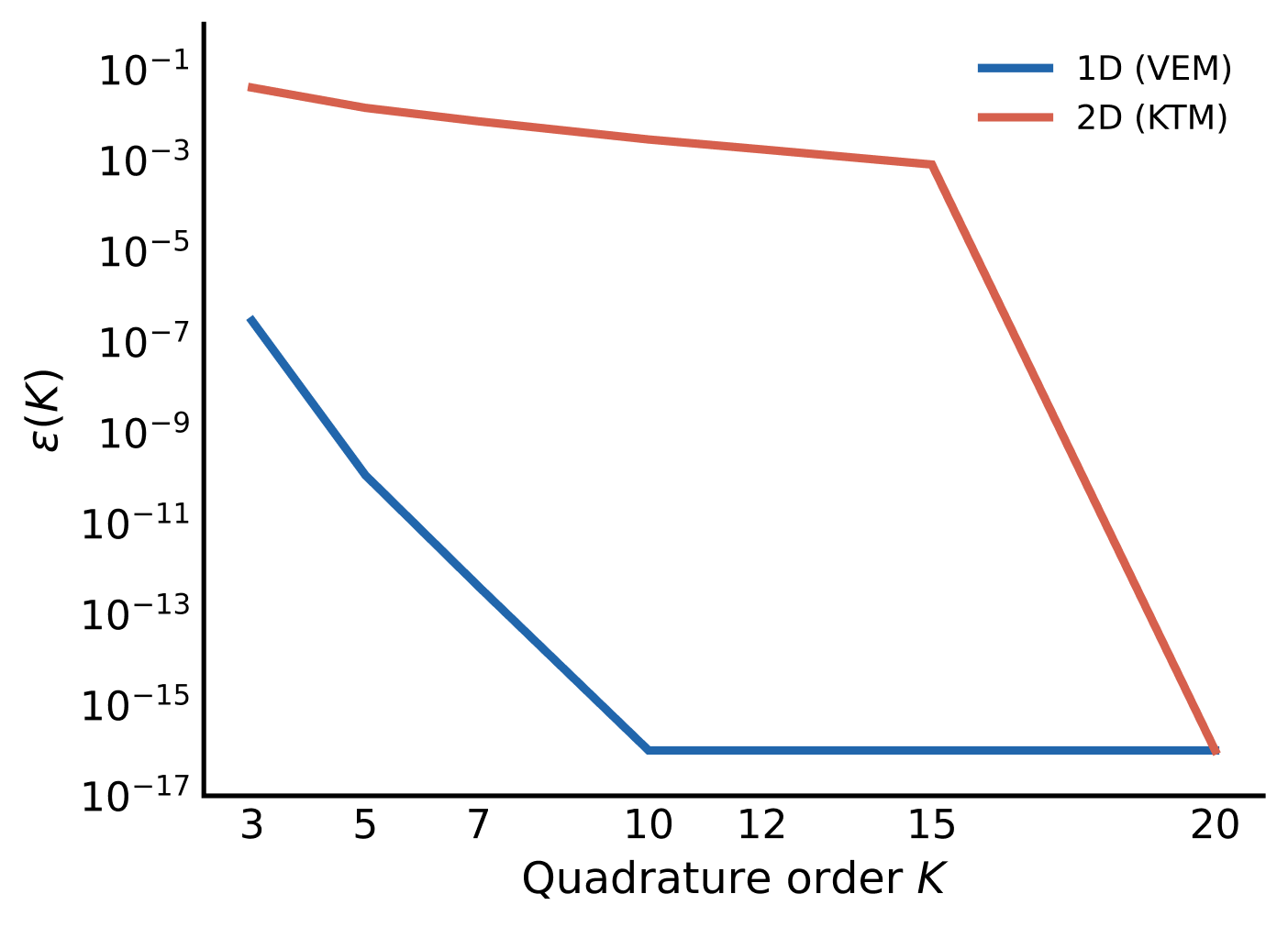}
\caption{Relative quadrature error $\varepsilon(K)$ on the M-R scenario. The univariate rule (VEM) reaches machine precision by $K{=}7$; the bivariate tensor-product rule (KTM) reaches $\varepsilon < 10^{-3}$ by $K{=}15$. All experiments use $K{=}10$.}
  \label{fig:gh_convergence}
\end{figure}

Figure~\ref{fig:gh_convergence} shows the quadrature relative error $\varepsilon(K)$ evaluated at converged parameters on the scenario from section~\ref{sec:sim}. All experiments in the main text use $K{=}10$, which is conservative for VEM and yields sub-percent accuracy for KTM.

\section{Machine Information}\label{sec:supp_machine}

All experiments were run in Python 3.12 on a Windows 11 workstation equipped with an Intel Core i7 CPU and an NVIDIA RTX 4070 SUPER GPU in \texttt{float64} precision. Total wall-clock time to convergence was $7$--$31$~s (VEM) and $34$--$37$~s (KTM) across kernel regimes at $N{=}45$, $T{=}75$, $C{=}12$. On cpu $74-164$~s (VEM)

\section{Variational Expectation Maximization}\label{sec:supp_vem}
As seen in section~\ref{sec:vem} the VEM is used to jointly estimate the parameters of the tuning curve, and Kronecker factors. In this section we derive the approximation for the evidence term, the closed form KL divergence term and the M-step. The E-step optimizes $(M_c, L_{\Psi_U}, L_{\Psi_V})$ jointly via Adam ($\eta = 0.05$, up to 100 iterations, relative tolerance $10^{-4}$) with Cholesky diagonal jittered to maintain positive definiteness.

\subsection{Evidence term}\label{sec:supp_evidence}
Because the Poisson likelihood factorizes over $(c, r, t)$, the expectation reduces to a sum of scalar integrals over the marginal of each entry of $N_c$. Under $q_\Psi$, each entry is Gaussian with mean $[M_c]_{rt}$ and marginal variance

\begin{equation}\label{eq:supp_margvar}
  \sigma^2_{rt} = [\Psi_U]_{rr}\,[\Psi_V]_{tt}.
\end{equation}

The expected log-likelihood therefore reduces to

\begin{equation}\label{eq:ell_decomp}
  \E[q_{\Psi}]{\log p(Y \mid N, \Theta)}
  = \sum_{c,r,t} \phi\!\left([M_c]_{rt},\,\sigma^2_{rt}\right),
\end{equation}

where each $\phi$ is a scalar Gaussian expectation evaluated by $K$-point GHQ (see \S\ref{sec:supp_ghq}).

\subsection{KL divergence term}\label{sec:supp_kl}
A short derivation of the closed form Kullback-Leibler divergence for the matrix-normal term is provided below.  
 
\begin{proposition}[KL]
\label{prop:kl}
Let $q_c = \MatNormal{R}{T}{M_c}{\Psi_U}{\Psi_V}$ and $p_c = \MatNormal{R}{T}{0}{U}{V}$. Then
\begin{align}\label{eq:kl}
  \sum_{c=1}^{C}\KL{q_c}{p_c}
  &= \frac{1}{2}\Bigg[
      \sum_{c=1}^{C}
        \tr\!\paren{V^{-1}M_c^\top U^{-1}M_c}
      + C\!\cdot\!
        \tr\!\paren{U^{-1}\Psi_U}
        \cdot
        \tr\!\paren{V^{-1}\Psi_V}
      \notag\\[-6pt]
      &\qquad\quad
      - CRT
      + CT\logdet U
      + CR\logdet V
      - CT\logdet\Psi_U
      - CR\logdet\Psi_V
    \Bigg].
\end{align}
\end{proposition}
\begin{proof}
The KL between two matrix-normals with the same dimensions is
(see, e.g., \cite{zotero-item-3623}, Ch.~2)
\begin{align}
  \KL{q_c}{p_c}
  &= \frac{1}{2}\Bigg[
       \tr\!\paren{(\Psi_V \otimes \Psi_U)(V^{-1} \otimes U^{-1})}
       + \tr\!\paren{V^{-1}M_c^\top U^{-1}M_c}
       \notag\\
       &\quad - RT
       + T\logdet U + R\logdet V
       - T\logdet\Psi_U - R\logdet\Psi_V
     \Bigg].
\end{align}
Applying the mixed-product Kronecker identity
$\tr(B^\top \otimes A) = \tr(A)\tr(B)$ to the first trace,
\begin{equation}
  \tr\!\paren{(\Psi_V \otimes \Psi_U)(V^{-1} \otimes U^{-1})}
  = \tr\!\paren{(V^{-1}\Psi_V)\otimes(U^{-1}\Psi_U)}
  = \tr\!\paren{U^{-1}\Psi_U}\,\tr\!\paren{V^{-1}\Psi_V}.
\end{equation}
Since $\Psi_U$ and $\Psi_V$ are shared across conditions, summing over
$c = 1,\ldots,C$ gives~\eqref{eq:kl}.
\end{proof}
The Cholesky-based expressions compute each term without forming inverses explicitly: $\tr(U^{-1}\Psi_U) = \fronorm{L_U^{-1}L_{\Psi_U}}^2$, $\sum_c\tr(V^{-1}M_c^\top U^{-1}M_c) = \sum_c\fronorm{L_U^{-1}M_c L_V^{-\top}}^2$.

\subsection{M-step}\label{sec:supp_mstep}
The M-step maximises the ELBO over $\Theta = \{U, V, f\}$ with $(M_{1:C}, \Psi_U, \Psi_V)$ fixed. Isolating the $U$-dependent terms of $-\KL{q}{p}$:
\begin{equation}\label{eq:U_obj}
  \ell(U)
  = -\frac{CT}{2}\logdet U
  - \frac{1}{2}\tr\!\paren{U^{-1}\mathbf{S}_U},
\end{equation}
where
\begin{equation}\label{eq:SU}
  \mathbf{S}_U
  = \sum_{c=1}^{C} M_c V^{-1} M_c^\top
    + C\cdot\tr\!\paren{V^{-1}\Psi_V}\cdot\Psi_U.
\end{equation}
Setting $\partial\ell/\partial U^{-1} = 0$ in a manner analogous to~\cite{chiquet2019}:
\begin{equation}\label{eq:Uhat_supp}
  \est{U}
  = \frac{1}{CT}\sum_{c=1}^{C} M_c V^{-1} M_c^\top
  + \frac{\tr\!\paren{V^{-1}\Psi_V}}{T}\,\Psi_U.
\end{equation}
The first term is a weighted sample covariance across conditions; the second is a
full-matrix correction proportional to $\Psi_U$, scaled by
$\tr(V^{-1}\Psi_V)/T$. Symmetrically, with $\est{U}$ substituted:
\begin{equation}\label{eq:Vhat_supp}
  \est{V}
  = \frac{1}{CN}\sum_{c=1}^{C} M_c^\top \est{U}^{-1} M_c
  + \frac{\tr\!\paren{\est{U}^{-1}\Psi_U}}{N}\,\Psi_V.
\end{equation}

\section{Kernel Tournament Method}\label{sec:supp_KTM}

\subsection{Stage 1 optimization.}\label{sec:sup_KTM_1}
Parameters $(\widehat{f}, L_U, L_V)$ are fit by minimizing $-\operatorname{CL}(\Theta;\bm{Y},\mathcal{S})$ via Adam ($\eta = 0.05$, $300$ steps) followed by L-BFGS with strong Wolfe line search ($20$ inner iterations, $\Delta_{\mathrm{tol}} = 10^{-7}$).

\subsection{Stage 2 optimization.}\label{sec:sup_KTM_2}
Each kernel candidate is scouted with $\widehat{f}$ fixed via Adam ($\eta = 0.04$, $300$ steps) followed by L-BFGS ($15$ inner iterations, $\Delta_{\mathrm{tol}} = 10^{-7}$). The winning family $k^{*}$ is selected by BIC-penalized negative composite likelihood evaluated on $\mathcal{S}$.

\subsection{Stage 3 optimization.}\label{sec:sup_KTM_3}
The selected kernel is refined with $\widehat{f}$ fixed via Adam ($\eta = 0.02$, up to $2000$ steps, $r_{\mathrm{tol}} = 10^{-5}$) followed by L-BFGS ($25$ inner iterations, $\Delta_{\mathrm{tol}} = 10^{-8}$), yielding the final estimates $(\widehat{\bm{U}}_{k^{*}}, \widehat{\bm{V}}_{k^{*}})$.

\section{Extended Simulated Results}\label{sec:supp_extended}
In this section we expand upon the benchmark results shown in section~\ref{sec:sim} to demonstrate performance across all combinations of the Kronecker factors for the proposed methods as well as across comparable methods.

\subsection{Proposed Method Benchmarks}\label{sec:supp_proposebench}

\begin{table}[H]
	\caption{Recovery of Kronecker factors and tuning curves across all four covariance
		regimes ($n=12$ seeds, mean $\pm$ SEM). The highlighted regime corresponds to the
		primary result reported in the main text.}
	\label{tab:supp_sim}
	\centering
	\small
	\begin{tabular}{llcc}
		\textit{prior structure} & & VEM & KTM \\
		\cline{3-4}
		\noalign{\vskip 1.5pt}
		\cline{3-4}
		& $\varepsilon(\widehat{\bm{U}})$ & $\mathbf{12.9\pm0.6}$ & $21.1\pm1.2$ \\
		\textit{LR$\times$LR} & $\varepsilon(\widehat{\bm{V}})$ & $45.4\pm0.4$ & $\mathbf{23.2\pm0.9}$ \\
		& $\varepsilon(\widehat{f})$ & $\mathbf{13.2\pm0.8}$ & $17.5\pm0.5$ \\
		\cline{3-4}
		\rowcolor{highlightyellow} & $\varepsilon(\widehat{\bm{U}})$ & $\mathbf{17.4\pm0.4}$ & $21.3\pm0.6$ \\
		\rowcolor{highlightyellow}\textit{LR$\times$M} & $\varepsilon(\widehat{\bm{V}})$ & $48.0\pm0.6$ & $\mathbf{4.5\pm0.8}$ \\
		\rowcolor{highlightyellow} & $\varepsilon(\widehat{f})$ & $12.6\pm0.8$ & $\mathbf{11.9\pm0.7}$ \\
		\cline{3-4}
		& $\varepsilon(\widehat{\bm{U}})$ & $15.9\pm0.4$ & $\mathbf{4.8\pm0.9}$ \\
		\textit{M$\times$LR} & $\varepsilon(\widehat{\bm{V}})$ & $40.5\pm0.2$ & $\mathbf{22.4\pm1.1}$ \\
		& $\varepsilon(\widehat{f})$ & $\mathbf{14.8\pm0.5}$ & $17.8\pm0.5$ \\
		\cline{3-4}
		& $\varepsilon(\widehat{\bm{U}})$ & $24.4\pm0.7$ & $\mathbf{5.0\pm0.7}$ \\
		\textit{M$\times$M} & $\varepsilon(\widehat{\bm{V}})$ & $40.1\pm0.5$ & $\mathbf{4.6\pm0.6}$ \\
		& $\varepsilon(\widehat{f})$ & $\mathbf{10.2\pm0.4}$ & $10.7\pm0.8$ \\
		\cline{3-4}
	\end{tabular}
\end{table}

Table~\ref{tab:supp_sim} reports $\varepsilon_F$ and $\rho$ for both methods across all four covariance regimes; the LR$\times$M regime (highlighted) is the primary regime discussed in the main text.

\subsection{Method Comparison Benchmarks}\label{sec:supp_methodbench}

\begin{table}[H]
\caption{Method comparison on simulated dataset from main paper with $N{=}45$, $C{=}12$, $T{=}75$. Bold = best per row. Mean over 12 seeds; $\pm$ 1 SEM.}
\label{tab:supp_compare}
\centering
\small
\begin{tabular}{llcccccc}
\textbf{prior structure}& & \multicolumn{6}{c}{\textbf{method mean relative error}} \\
& & VEM & KTM & GPLVM & Goris & GPFA & MLE (GH) \\
\cline{3-8}
\noalign{\vskip 1.5pt}
\cline{3-8}
   & $\varepsilon(\widehat{\mathbf{U}})$ & $\mathbf{12.9\pm0.6}$ & $21.1\pm1.2$ & $96.6\pm0.0$ & $63.1\pm0.6$ & $101.4\pm1.4$ & $93.5\pm0.0$ \\
  \textit{LR-LR} & $\varepsilon(\widehat{\mathbf{V}})$ & $45.4\pm0.4$ & $\mathbf{23.2\pm0.9}$ & $75.6\pm0.0$ & $27.7\pm0.7$ & $118.9\pm0.1$ & $71.1\pm0.0$ \\
   & $\varepsilon(\hat{f})$ & $\mathbf{13.2\pm0.8}$ & $17.5\pm0.5$ & $36.1\pm0.4$ & $59.8\pm0.6$ & $59.8\pm0.6$ & $16.1\pm0.5$ \\
\cline{3-8}
 \rowcolor{highlightyellow}  & $\varepsilon(\widehat{\mathbf{U}})$ & $\mathbf{17.4\pm0.4}$ & $21.3\pm0.6$ & $96.6\pm0.0$ & $36.0\pm0.8$ & $129.6\pm1.2$ & $93.3\pm0.0$ \\
 \rowcolor{highlightyellow} \textit{LR-Mat} & $\varepsilon(\widehat{\mathbf{V}})$ & $48.0\pm0.6$ & $\mathbf{4.5\pm0.8}$ & $83.4\pm0.2$ & $45.1\pm0.8$ & $48.3\pm0.5$ & $90.8\pm0.0$ \\
 \rowcolor{highlightyellow}  & $\varepsilon(\hat{f})$ & $12.6\pm0.8$ & $\mathbf{11.9\pm0.7}$ & $17.5\pm0.4$ & $47.0\pm0.6$ & $47.0\pm0.6$ & $10.7\pm0.5$ \\
\cline{3-8}
   & $\varepsilon(\widehat{\mathbf{U}})$ & $15.9\pm0.4$ & $\mathbf{4.8\pm0.9}$ & $63.5\pm0.7$ & $48.9\pm0.8$ & $142.1\pm5.4$ & $92.4\pm0.0$ \\
  \textit{Mat-LR} & $\varepsilon(\widehat{\mathbf{V}})$ & $40.5\pm0.2$ & $\mathbf{22.4\pm1.1}$ & $75.6\pm0.0$ & $24.1\pm0.5$ & $95.3\pm0.3$ & $71.0\pm0.0$ \\
   & $\varepsilon(\hat{f})$ & $\mathbf{14.8\pm0.5}$ & $17.8\pm0.5$ & $38.3\pm0.7$ & $66.7\pm0.8$ & $66.7\pm0.8$ & $17.4\pm0.4$ \\
\cline{3-8}
   & $\varepsilon(\widehat{\mathbf{U}})$ & $24.4\pm0.7$ & $\mathbf{5.0\pm0.7}$ & $82.9\pm0.1$ & $43.0\pm1.0$ & $174.9\pm1.6$ & $92.5\pm0.0$ \\
  \textit{Mat-Mat} & $\varepsilon(\widehat{\mathbf{V}})$ & $40.1\pm0.5$ & $\mathbf{4.6\pm0.6}$ & $78.5\pm0.1$ & $38.5\pm0.3$ & $49.4\pm0.9$ & $90.8\pm0.0$ \\
   & $\varepsilon(\hat{f})$ & $10.2\pm0.4$ & $10.7\pm0.8$ & $18.2\pm0.6$ & $49.7\pm0.7$ & $49.7\pm0.7$ & $\mathbf{9.5\pm0.3}$ \\
\cline{3-8}
\end{tabular}
\end{table}

Table~\ref{tab:supp_compare} reports $\varepsilon_F$ and $\rho$ between our proposed methods and related methods across all four covariance regimes; the highlighted is the primary regime discussed in the main text.

\section{Subspace Alignment}\label{sec:supp_subspace}

We computed the subspace alignment between the neural activity and the latent gain by using the projection method from~\cite{elsayed_reorganization_2016}. 

\begin{equation}
\frac{\text{Tr}(\bm{E}^T \bm{C} \bm{E})}{\sum_{i=1}^{10} \lambda_i(\bm{C})},
\end{equation}

where E is a matrix of the first 10 eigenvectors of $\bm{U}$, $\bm{C}$ is the neural activity covariance matrix, and $\lambda_i(\bm{C})$ is the $\text{i}^{\text{th}}$ eigenvalue value of $\bm{C}$.


\end{document}